\documentclass[11pt,reqno,letterpaper]{amsart}
\usepackage{color}
\usepackage[colorlinks=true, allcolors=blue,backref=page]{hyperref}
\usepackage{amsmath, amssymb, amsthm}
\usepackage{mathrsfs}

\usepackage{mathtools}
\usepackage[noabbrev,capitalize,nameinlink]{cleveref}
\crefname{equation}{}{}
\usepackage{fullpage}
\usepackage[noadjust]{cite}
\usepackage{graphics}
\usepackage{pifont}
\usepackage{tikz}
\usepackage{bbm}
\usepackage[T1]{fontenc}

\usetikzlibrary{arrows.meta}

\usepackage{environ}
\usepackage{framed}
\usepackage{url}
\usepackage[linesnumbered,ruled,vlined]{algorithm2e}
\usepackage[noend]{algpseudocode}
\usepackage[labelfont=bf]{caption}
\usepackage{cite}
\usepackage{framed}
\usepackage[framemethod=tikz]{mdframed}
\usepackage{appendix}
\usepackage{graphicx}
\usepackage[textsize=tiny]{todonotes}
\usepackage{tcolorbox}
\usepackage{enumerate}
\allowdisplaybreaks[1]
\usepackage{enumerate}
\usepackage{stmaryrd}

\usepackage[margin=1in]{geometry}
\usepackage[shortlabels]{enumitem}
\crefformat{enumi}{#2#1#3}
\crefrangeformat{enumi}{#3#1#4 to~#5#2#6}
\crefmultiformat{enumi}{#2#1#3}
{ and~#2#1#3}{, #2#1#3}{ and~#2#1#3}

\usepackage{float}

\usepackage[normalem]{ulem}

\usepackage[shortlabels]{enumitem}
\crefformat{enumi}{#2#1#3}
\crefrangeformat{enumi}{#3#1#4 to~#5#2#6}
\crefmultiformat{enumi}{#2#1#3}
{ and~#2#1#3}{, #2#1#3}{ and~#2#1#3}

\DeclareSymbolFont{symbolsC}{U}{pxsyc}{m}{n}
\SetSymbolFont{symbolsC}{bold}{U}{pxsyc}{bx}{n}
\DeclareFontSubstitution{U}{pxsyc}{m}{n}
\DeclareMathSymbol{\medcircle}{\mathbin}{symbolsC}{7}

\crefname{algocf}{Algorithm}{Algorithms}

\crefname{equation}{}{} 
 
\AtBeginEnvironment{appendices}{\crefalias{section}{appendix}} 

\usepackage[color,final]{showkeys}

\colorlet{refkey}{orange!20}
\colorlet{labelkey}{blue!30}

\crefname{algocf}{Algorithm}{Algorithms}

\numberwithin{equation}{section}
\newtheorem{theorem}{Theorem}[section]
\newtheorem{proposition}[theorem]{Proposition}
\newtheorem{lemma}[theorem]{Lemma}
\newtheorem{claim}[theorem]{Claim}

\crefname{claim}{Claim}{Claims}

\newtheorem{corollary}[theorem]{Corollary}

\newtheorem*{question*}{Question}

\theoremstyle{definition}
\newtheorem{definition}[theorem]{Definition}

\newtheorem{question}[theorem]{Question}
\newtheorem*{definition*}{Definition}
\newtheorem{example}[theorem]{Example}

\theoremstyle{remark}
\newtheorem*{remark}{Remark}

\newlist{enumthm}{enumerate}{1}
\setlist[enumthm]{label=\textup{(\roman*)},ref=\thethm(\roman*)}
\Crefname{enumthmi}{Theorem}{Theorems}

\newcommand{\R}{\mathbb R}

\newcommand{\mb}{\mathbb}

\newcommand{\tw}{\operatorname{tw}}
\newcommand{\elltwo}{\ell_2}
\newcommand{\norm}[1]{\left\lVert #1 \right\rVert}

\newcommand{\eps}{\varepsilon}

\let\originalleft\left
\let\originalright\right
\renewcommand{\left}{\mathopen{}\mathclose\bgroup\originalleft}
\renewcommand{\right}{\aftergroup\egroup\originalright}

\allowdisplaybreaks

\newif\ifpublic

\ifpublic

\newcommand{\ignore}[1]{}

\else

\fi

\title{On graphically local versions of metric embeddings}

\author{Vishesh Jain}
\address{Department of Mathematics, Statistics, and Computer Science,
University of Illinois Chicago, Chicago, IL 60607, USA}
\email{visheshj@uic.edu}

\author{Duan Tu}
\address{Department of Mathematics, Statistics, and Computer Science,
University of Illinois Chicago, Chicago, IL 60607, USA}
\email{dtu4@uic.edu}

\begin{document}

\begin{abstract}
We consider the problem of \emph{graphically local} metric embedding, i.e.~embedding points from an arbitrary finite metric space into a target metric space while preserving, up to a small distortion, only a subset of the pairwise distances specified by a bounded degree graph $G$. 

We provide a general reduction showing that, in many cases, this is no easier than embedding the points while approximately preserving \emph{all} pairwise distances. As an illustration of our general reduction, we show that there exists a Euclidean metric space $X$ on $n$ points along with a graph $G = (X,E)$ of maximum degree $3$ such that any embedding of $X$ into $\ell_2^m$ which only preserves distances specified by $E$ up to a relative error of $(1+\eps)$ must satisfy $m = \Omega(\eps^{-2}\log n)$. 

Our lower bound matches the upper bound on the dimension coming from the Johnson-Lindenstrauss lemma for approximately preserving all pairwise distances;  previously, such a lower bound was known only for the class of noncontracting embeddings [Schechtman-Shraibman, Discrete \& Computational Geometry, 2009]. Moreover, the condition that the maximum degree of the graph is $3$ is best possible: for graphs $G$ of maximum degree $2$ (or more generally, treewidth at most $2$), any metric space embeds $G$-isometrically into any two-dimensional normed space.  
\end{abstract}

\maketitle

\section{Introduction}

Dimension reduction is a cornerstone of modern data analysis and theoretical computer science, providing methods to represent high-dimensional data in lower-dimensional spaces while preserving essential geometric properties. The Johnson-Lindenstrauss (JL) lemma \cite{johnson1984extensions} is a seminal result, establishing that any set of $n$ points in Euclidean space can be embedded into $\ell_2^{O(\log n/\eps^2)}$ (the Euclidean space in $O(\log n / \eps^2)$ dimensions) while preserving all pairwise Euclidean distances up to a $(1+\eps)$ distortion. It was shown by Larsen and Nelson that this bound is tight up to constant factors for $\eps \in (n^{-0.499}, 1)$ (see \cite[Theorem~2]{LarsenNelson17} for a more general statement). Another seminal result in the theory of metric embeddings is Bourgain's embedding theorem \cite{bourgain1985lipschitz}, which asserts that every metric space on $n$ points can be embedded into Euclidean space with distortion $O(\log n)$; this is known to be tight up to constant factors (see,~e.g.,~\cite{matouvsek2013lecture}). Here, we say that an embedding $f$ from a metric space $(X,d_X)$ to a metric space $(Y,d_Y)$ has distortion at most $D \geq 1$ if there exists $r > 0$ such that
\[r\cdot d_X(x,y) \leq d_Y(f(x), f(y)) \leq D\cdot r\cdot d_X(x,y) \qquad \forall x,y \in X.\]

\paragraph{\bf Local dimension reduction} In numerous applications, particularly those dealing with massive or inherently complex datasets like social networks or biological data, preserving the full distance information is often unnecessary and may be computationally expensive without incurring high distortion. Instead, maintaining the \emph{local} structure of the data – the geometry and relationships of ``neighboring'' points – is of primary interest. This motivates the formal study of \emph{local dimension reduction}, initiated by Abraham, Bartal, and Neiman \cite{abraham2007local} (see also the expanded journal version \cite{abraham2015local}). They proposed several novel formalisms \cite[Definition~1]{abraham2007local} to capture the notion of local distortion that are based on the metric space's intrinsic structure, such as \emph{$k$-local distortion}: an embedding $f$ from a metric space $(X, d_X)$ to a metric space $(Y,d_Y)$ is said to have $k$-local distortion $\alpha$ if $d_Y(f(u), f(v)) \leq d_X(u,v)$ for all $u,v \in X$ and $d_Y(f(u), f(v))\geq d_X(u,v)/\alpha$ for every $u$ and every $v$ which is among the $k$-nearest neighbors of $u$. 

The results of \cite{abraham2007local} demonstrated that for their metric-based notions of local distortion, it is often possible to achieve embeddings with distortion or dimension dependent only on the locality parameter (e.g.~$k$, in the notion of $k$-local distortion above) rather than the total number of points $n$. For instance, improving their result \cite[Theorem~2]{abraham2007local}, they showed in follow up work \cite[Theorem~1]{abraham2009low} that any metric space (on $n$ points) can be embedded into $\ell_p^{O(e^p \log^2 k)}$ with $k$-local distortion $O(\log k/p)$ (for any $k\leq n$, $1\leq p \leq \log k$). In the low distortion setting of the JL lemma, they showed \cite[Theorem~2]{abraham2009low} that for any $\eps > 0$ and $p\geq 1$, an ultrametric space admits an embedding into $\ell_p^{O(\log k/\eps^3)}$ with $k$-local distortion ($1+\eps$). 

A natural question (asked in \cite[Section~11]{abraham2007local}) is whether the ultrametricity assumption in the above result can be removed; specifically, is there a $k$-local version of the JL lemma, i.e.~ can any finite set of points in $\ell_2$ be embedded into $\ell_2^{O(\log k/\eps^2)}$ with distortion $(1+\eps)$? This was answered in the negative by Schechtman and Shraibman \cite{schechtman2009lower}, who constructed a set of $n+1$ points $X\subseteq \mb{R}^n$ such that any embedding $f: X \to \ell_2^m$ with $3$-local distortion $(1+\eps)$ (for a sufficiently small constant $\eps > 0$) must have dimension $m = \Omega (\log n)$ \cite[Theorem~9]{schechtman2009lower} (see also \cite[Theorem~8]{schechtman2009lower} which obtains a lower bound with near optimal $\eps$ dependence under more requirements on the embedding).

\paragraph{\bf Graphically local dimension reduction} The focus of this article is a different, graphical notion of local dimension reduction, which was studied by Schechtman and Shraibman \cite[Section~5]{schechtman2009lower}. 

\begin{definition}
    \label{def:G-local-distortion}
Let $(X, d_X)$ and $(Y,d_Y)$ be metric spaces. Let $G = (X,E)$ be a graph whose vertices are points of $X$. We say that a (not-necessarily injective) map $f: (X,d_X) \to (Y,d_Y)$ is a $G$-local $D$-embedding, where $D\geq 1$ is a real number, if there is a real number $r > 0$ such that 
\begin{equation}
\label{eq:distortion-local}
r \cdot d_X(x,y) \leq d_Y(f(x),f(y)) \leq D\cdot r \cdot d_X(x,y) \qquad \forall \{x,y\} \in E. 
\end{equation}
The infimum of the numbers $D\geq 1$ such that $f$ is a $G$-local $D$-embedding is called the $G$-local distortion of $f$.
\end{definition}

\begin{remark}
In the case that $G$ is the complete graph on $X$, this reduces to the usual notion of (global) distortion. Also, $k$-local distortion essentially corresponds to the case when $G$ is the $k$-nearest neighbour graph on $X$. 
\end{remark}

Perhaps the most natural choice of locality parameter in this setting is the maximum degree $\Delta$ of the graph $G$. In a similar spirit as the work of Abraham, Bartal, and Neiman \cite{abraham2007local}, one can ask whether it is possible to achieve embeddings with $G$-local distortion or dimension dependent only on this locality parameter $\Delta$. The possibility of this is suggested by the following example, highlighted in both \cite{abraham2009low} and \cite{schechtman2009lower}.  

\begin{example}
\label[example]{example:equilateral}
    Let $(X,d_X)$ be the set of points $X = \{e_1,\dots, e_n\} \subseteq \mb{R}^n$ endowed with the Euclidean metric. Alon \cite{alon2003problems} showed that any embedding $f: (X, d_X) \to \ell_2^m$ with distortion $(1+\eps)$ must satisfy $m = \Omega(\log n/(\log (1/\eps)\eps^2))$. However, for any $G = (X,E)$ of maximum degree $\Delta$, one may construct a $G$-local embedding $g: X \to \ell_2^{O(\log \Delta/\eps^2)}$ with distortion $(1+\eps)$ as follows. First, since $G$ has maximum degree at most $\Delta$, we may greedily find a map $h: X \to \{e_1,\dots, e_{\Delta+1}\} \subseteq \ell_2^{\Delta+1}$ such that $h(u)\neq h(v)$ if $\{u,v\} \in E$. The crucial property of $h$ is that it is a $G$-local isometry, i.e.~for any $\{u,v\}\in E$, $d_X(u,v) = \|h(u)-h(v)\|_2$. Now, we compose $h$ with a Johnson-Lindenstrauss map from $\ell_2^{\Delta+1}$ to $\ell_2^{O(\log \Delta/\eps^2)}$ to obtain the desired embedding $g$. 
\end{example}

Since the set in \cref{example:equilateral} is a standard ``hardness'' example for the JL lemma, one may ask whether there is a $G$-local version of the JL lemma depending on the maximum degree $\Delta$. For instance, this was asked in \cite{jain2023dimension}, where it was observed that if this were true, the embedding achieving this result would necessarily have to be non-linear. To the best of our knowledge, the only lower bound known for this problem is due to Schechtman and Shraibman. In \cite[Theorem~11]{schechtman2009lower}, they show that for $X$ as in \cref{example:equilateral} and for any $d$-regular $G = (X,E)$ with $d\geq 3$ and second eigenvalue bounded by $d/2$, any $G$-local $(1+\eps)$-embedding $f: X \to \ell_2^m$ \emph{which is non-contracting} (i.e.~satisfies $\|f(x)-f(y)\|_2 \geq (1-\eps)\|x-y\|_2$ for \emph{all} $x,y \in X$) must satisfy $m = \Omega(\log n)$. Of course, the non-contracting assumption precludes the embedding of \cref{example:equilateral}, which achieves $m = O(\log \Delta/\eps^2)$.

\paragraph{\bf Our contribution}  We show that in many settings of interest, including those of the JL lemma and Bourgain's embedding theorem, $G$-local embeddings do not provide \emph{any} asymptotic saving in the dimension/distortion over global embeddings. This is in sharp contrast to metrically local embeddings (as in \cite{abraham2007local}) for which such savings are possible in high distortion regimes, as discussed above. 

Our main technical contribution is a general reduction (\cref{thm:gen-lb}), which takes as input a metric space $(X,d_X)$ and constructs a metric extension $(Z,d_Z)$ along with a graph $G = (Z,E)$ of maximum degree $3$ such that any map $f: (Z, d_Z) \to (Y,d_Y)$ which (approximately) preserves distances for points in $Z$ connected by an edge must necessarily (approximately) preserve \emph{all} pairwise distances among points in $X$ (in fact, there is a version of this statement for any distortion $D$). By applying this reduction with $(X, d_X)$ being a known hard instance for some metric embedding problem, we are able to construct hard instances for $G$-local metric embeddings. We illustrate this with two applications: \cref{cor:local-JL}, which is a lower bound for $G$-local JL which is optimal up to a constant in the stated range of $\eps$, and \cref{cor:l1-embedding}, which is a lower bound for $G$-local embeddings of arbitrary finite metric spaces into Euclidean space, which is again optimal up to a constant. We also highlight some open problems.

The condition that the graph in our reduction has maximum degree $3$ is best possible. In \cref{sec:upper bounds}, we show that if $G$ has maximum degree at most $2$ (or more generally, treewidth at most $2$), then any metric space embeds $G$-isometrically into $\mb{R}^2$ equipped with any norm.

\subsection*{Acknowledgments}
The authors used generative AI tools for assistance with figures and proofreading. The authors thank Sidhanth Mohanty for useful discussions and an anonymous referee for suggesting a simplification that led to the unified formulation of \cref{thm:gen-lb} appearing here. V.J.~is partially supported by NSF CAREER award DMS-2237646.

\section{Upper bounds for graphs of bounded treewidth}
\label{sec:upper bounds}

Treewidth is a fundamental graph parameter that measures how ``tree-like'' a graph is.
Graphs of small treewidth admit efficient dynamic-programming algorithms for many problems
that are intractable on general graphs. This phenomenon is central in parameterized complexity:
a large class of NP-hard graph problems become fixed-parameter tractable (FPT) when
parameterized by treewidth (or related parameters). Concretely, many optimization and
constraint satisfaction problems can be solved in time
\[
f(k)\cdot n^{O(1)} \qquad\text{on graphs of treewidth at most }k,
\]
where \(n\) is the number of vertices and \(f\) depends only on \(k\).
This includes (among many others) \textsc{Vertex Cover}, \textsc{Independent Set}, and \textsc{Hamilton Cycle} (see, e.g., \cite{cygan2015parameterized}).

Motivated by this, we study $G$-local metric embeddings for graphs $G$ of bounded treewidth. 

\begin{definition}[Tree decomposition and treewidth]\label{def:tw-bags}
A {tree decomposition} of a graph \(G=(V,E)\) is a pair \((T,\mathcal{B})\) where
\(T\) is a tree and \(\mathcal{B}=\{B_t\subseteq V: t\in V(T)\}\) is a family of subsets
(called {bags}) such that:
\begin{enumerate}[label=(\roman*)]
\item \(\bigcup_{t\in V(T)} B_t = V\).
\item For every edge \(\{u,v\}\in E\), there exists \(t\in V(T)\) such that \(\{u,v\}\subseteq B_t\).
\item ({Running intersection property}) For each vertex \(v\in V\), the set
\(\{t\in V(T): v\in B_t\}\) induces a connected subtree of \(T\).
\end{enumerate}
The {width} of \((T,\mathcal{B})\) is \(\max_{t\in V(T)}(|B_t|-1)\).
The {treewidth} \(\tw(G)\) is the minimum width over all tree decompositions of \(G\).
\end{definition}
\begin{example}
    For any tree $T$, $\tw(T) =1$. For the cycle graph $C_n$ on $n\geq 3$ vertices, $\tw(C_n) = 2$. For the complete graph $K_n$ on $n$ vertices, $\tw(K_n) = n-1$. 
\end{example}

Our first result shows that any Euclidean metric $G$-locally isometrically embeds into $\ell_2^{k}$, provided that $\tw(G)\leq k$. 

\begin{theorem}\label{thm:twk-euclid}
Let \(G=(X,E)\) be a graph with \(\tw(G)\le k\).
Assume \((X,d)\) is a Euclidean metric, i.e.\ there exists an embedding \(\iota:X\to \elltwo\) such that
\(d(x,y)=\norm{\iota(x)-\iota(y)}_2\) for all \(x,y\in X\).
Then there exists a map \(f:X\to\R^k\) such that
\[
\norm{f(x)-f(y)}_2 = d(x,y)\qquad \forall \{x,y\}\in E.
\]
In words, every Euclidean metric space admits a \(G\)-local isometry into \(\R^k\) whenever \(\tw(G)\le k\).
\end{theorem}

For graphs of treewidth at most $2$, a much stronger statement is true. 

\begin{theorem}\label{thm:tw2}
Let $G=(X,E)$ be a graph with $\tw(G)\le 2$, and let $(X,d)$ be an arbitrary finite metric space.
Let $(V,\|\cdot\|)$ be any two-dimensional normed space.
Then there exists a map $f:X\to V$ such that
\[
\|f(x)-f(y)\| = d(x,y)\qquad \forall \{x,y\}\in E.
\]
In words, every finite metric space admits a $G$-local isometry into {any} normed plane whenever $\tw(G)\le 2$.
\end{theorem}

\begin{remark}
    The condition $\tw(G)\leq 2$ in \cref{thm:tw2} is best possible for $G$-local isometry, even if the target is the infinite-dimensional space $\ell_2$. For a counterexample with treewidth $3$, let $X=\{0,1,2,3\}$ and define the metric $d$ by
\[
d(0,i)=1 \quad (i=1,2,3),
\qquad
d(i,j)=2 \quad (1\le i<j\le 3).
\]
This is the shortest--path metric of a star with center $0$. Let $G=K_4$, so $G$-local distortion
coincides with ordinary distortion. Let $f:X\to \ell_2$ have distortion $D$. After rescaling and translation, assume $f(0)=0$ and
\[
\|f(i)\|_2\le D \quad (i=1,2,3),
\qquad
\|f(i)-f(j)\|_2\ge 2 \quad (1\le i<j\le 3).
\]
The upper bounds give
\[
\sum_{1\leq i<j \leq 3}\|f(i)-f(j)\|_2^2
=
3\sum_{i=1}^3 \|f(i)\|_2^2
-
\Bigl\|\sum_{i=1}^3 f(i)\Bigr\|^2
\le
3\sum_{i=1}^3 \|f(i)\|_2^2
\le
9D^2,
\]
while the lower bounds give
\[
\sum_{i<j}\|f(i)-f(j)\|_2^2 \ge 3\cdot 2^2 = 12.
\]
Thus $12\le 9D^2$, so $D\ge 2/\sqrt{3}$.
\end{remark}

\begin{question}
In the previous remark, we showed that $\mathrm{tw}(G)\leq 2$ is best possible for a $G$-local isometry into $\ell_2$. It is an interesting question whether arbitrary metric spaces can be $G$-locally embedded into $\ell_2$ with distortion depending only on $\tw(G)$ (rather than on the number of points). In other words, is there a ``treewidth version'' of Bourgain's embedding theorem?    
\end{question}

The proofs of both of these theorems will rely on an equivalent characterization of treewidth in terms of $k$-trees. 

\begin{definition}[\(k\)-tree]\label[definition]{def:ktree}
Fix an integer \(k\ge 1\).
A graph \(H\) is a {\(k\)-tree} if it can be constructed by the following process:
\begin{enumerate}[label=(\roman*)]
\item Start with a clique on \(k+1\) vertices.
\item Repeatedly add a new vertex \(v\) whose neighborhood among the already-present vertices
is exactly a clique of size \(k\) (i.e.\ \(v\) is made adjacent to all vertices of some existing \(k\)-clique).
\end{enumerate}
\end{definition}

\begin{definition}[Partial \(k\)-tree]\label{def:partialktree}
A graph \(G\) is a {partial \(k\)-tree} if \(G\) is a subgraph of some \(k\)-tree on the same vertex set.
\end{definition}

The following equivalence is classical. 

\begin{theorem}[see, e.g.,~{\cite[Chapter~12]{diestel2025graph}}]
    For any integer $k\geq 1$ and any graph $G$, $\tw(G)\leq k$ if and only if $G$ is a partial $k$-tree. 
\end{theorem}

\begin{proof}[Proof of \cref{thm:twk-euclid}]
Since \(\tw(G)\le k\), \(G\) is a partial \(k\)-tree.
Thus there exists a \(k\)-tree \(H=(X,E_H)\) such that \(E\subseteq E_H\).
It suffices to construct an \(H\)-local isometry $f: X \to \mb{R}^k$. 
Fix a construction order \(v_1,\dots,v_n\) of the \(k\)-tree \(H\) as in \cref{def:ktree}:
the first \(k+1\) vertices form a \((k+1)\)-clique, and each subsequent vertex \(v_i\) (for \(i>k+1\))
has its earlier neighbors forming a \(k\)-clique \(S_i\).

We will construct the $H$-local isometry \(f\) by induction on \(i\).

\emph{Base case.}
Let \(C_k=\{v_1,\dots,v_{k+1}\}\).
In the given Euclidean realization \(\iota\) of the metric space $(X,d)$, the set \(\iota(C_k)\subseteq\elltwo\) lies in an affine subspace
of dimension at most \(k\) (since it has \(k+1\) points).
Choose an isometry from that affine subspace into \(\R^k\) and define \(f\) on \(C_k\) accordingly.
Then for all \(u,v\in C_k\), \(\norm{f(u)-f(v)}_2=d(u,v)\), so that $f$ is an $H$-local isometric embedding of $(C_k, d)$ into $\mb{R}^k$. 

\emph{Inductive step.}
Let $i > k+1$. Assume that \(f\) has already been defined on \(C_{i-1} = \{v_1,\dots,v_{i-1}\}\) and that $f$ is an $H$-local isometric embedding of $(C_{i-1}, d)$ into $\mb{R}^k$. 

Let \(v_i\) be the next vertex and let \(S_i=\{u_1,\dots,u_k\}\subseteq\{v_1,\dots,v_{i-1}\}\) be its
neighbor \(k\)-clique (so \(\{v_i,u_j\}\in E_H\) for all \(j \in [k]\)).

Consider the \((k+1)\)-point set \(S_i\cup\{v_i\}\) in \(\elltwo\) via \(\iota\).
These \(k+1\) points lie in an affine subspace of dimension at most \(k\); hence there exists a map
\(\psi_i:S_i\cup\{v_i\}\to\R^k\) such that
\(\norm{\psi_i(a)-\psi_i(b)}_2=d(a,b)\) for all \(a,b\in S_i\cup\{v_i\}\).

Moreover, for $a,b \in S_i$, it follows from the inductive hypothesis that
\[\|f(a) - f(b)\|_2 = d(a,b).\]

It follows that the labeled configurations \(f(S_i)\) and \(\psi_i(S_i)\) of $k$ points are congruent in \(\R^k\),
so there exists a Euclidean isometry (a composition of an orthogonal map and a translation) \(T_i:\R^k\to\R^k\) 
such that
\[
T_i\bigl(f(a)\bigr)=\psi_i(a)\qquad \text{for each }a \in S_i.
\]
Extend $f$ to $C_i = \{v_1,\dots, v_i\}$ by defining
\[
f(v_i)\ :=\ T_i^{-1}\bigl(\psi_i(v_i)\bigr).
\]
Then for every \(u_j\in S_i\),
\[
\norm{f(v_i)-f(u_j)}_2
= \norm{T_i(f(v_i)) - T_i(f(u_j))}_2
= \norm{\psi_i(v_i)-\psi_i(u_j)}_2
= d(v_i,u_j),
\]
so all new edges \(\{v_i,u_j\}\) are realized isometrically. This completes the inductive step.
\end{proof}

The proof of \cref{thm:tw2} follows the same template as the proof of \cref{thm:twk-euclid}. We will also need the following simple geometric lemma.  

\begin{lemma}\label[lemma]{lem:sphere-intersection-normed-plane}
Let $(V,\|\cdot\|)$ be a normed space of dimension at least $2$.
Fix $a,b\in V$ and radii $r_1,r_2\ge 0$, and set $t:=\|a-b\|$.
Then the spheres
\[
S(a,r_1):=\{x\in V:\|x-a\|=r_1\},
\qquad
S(b,r_2):=\{x\in V:\|x-b\|=r_2\}
\]
intersect if and only if
\[
|r_1-r_2|\ \le\ t\ \le\ r_1+r_2.
\]
\end{lemma}

\begin{proof}
The ``only if'' direction is immediate: if $x\in S(a,r_1)\cap S(b,r_2)$, then
\[
t=\|a-b\|\le \|a-x\|+\|x-b\|=r_1+r_2
\]
and by the reverse triangle inequality,
\[
t=\|a-b\|\ge \bigl|\|a-x\|-\|x-b\|\bigr| = |r_1-r_2|.
\]

For the converse, assume $|r_1-r_2|\le t\le r_1+r_2$.
By translation, we may assume $a=0$ and $b=w$, where $\|w\|=t$.
If $t=0$, then the assumed inequalities imply $r_1=r_2$, and the conclusion is immediate. We may therefore assume that $t>0$.
Consider the sphere $S(0,r_1)$ and the continuous function $\phi(u):=\|u-w\|$ on $S(0,r_1)$.

First note that $S(0,r_1)$ is connected (indeed path-connected) since $\dim V\ge 2$.
Hence $\phi(S(0,r_1))$ is a connected subset of $\R$, and therefore an interval.

For every $u\in S(0,r_1)$, the triangle inequality gives
\[
\phi(u)=\|u-w\|\le \|u\|+\|w\|=r_1+t,
\]
and the reverse triangle inequality gives
\[
\phi(u)=\|u-w\|\ge \bigl|\|w\|-\|u\|\bigr|=|t-r_1|.
\]
Both bounds are achieved by the points
\[
u_+:=\frac{r_1}{t}w,\qquad u_-:=-\frac{r_1}{t}w
\]
satisfy $\|u_\pm\|=r_1$ and
\[
\phi(u_+)=\left\|\left(\frac{r_1}{t}-1\right)w\right\|=|r_1-t|,
\qquad
\phi(u_-)=\left\|\left(-\frac{r_1}{t}-1\right)w\right\|=r_1+t.
\]
Therefore $\phi(S(0,r_1))$ is an interval containing both endpoints, hence
\[
\phi(S(0,r_1))=[\,|t-r_1|,\ t+r_1\,].
\]
Since $r_2\in[|t-r_1|,t+r_1]$ by the assumed inequalities, there exists $u\in S(0,r_1)$ with
$\|u-w\|=r_2$, i.e.\ $u\in S(0,r_1)\cap S(w,r_2)$, which translates back to
$S(a,r_1)\cap S(b,r_2)\neq\emptyset$.
\end{proof}

We can now prove \cref{thm:tw2}.

\begin{proof}[Proof of \cref{thm:tw2}]
Since $\tw(G)\le 2$, $G$ is a partial $2$-tree.
Thus there exists a $2$-tree $H=(X,E_H)$ with $E\subseteq E_H$.
It suffices to construct an $H$-local isometry into $(V,\|\cdot\|)$.

We induct on the construction of the $2$-tree $H$.
A $2$-tree starts with a triangle and repeatedly adds a new vertex adjacent to both endpoints of an existing edge.

\emph{Base step.}
Let the initial triangle be on vertices $\{a,b,c\}$.
Set $f(a)=0$.
Choose any vector $u\in V$ with $\|u\|=1$ and set $f(b)=d(a,b)\,u$.
We now place $f(c)$ so that
\[
\|f(c)-f(a)\|=d(c,a)\quad\text{and}\quad \|f(c)-f(b)\|=d(c,b).
\]
By \cref{lem:sphere-intersection-normed-plane} (with centers $f(a),f(b)$ and radii $d(c,a),d(c,b)$),
such a point exists because the triangle inequalities for the metric $d$ imply
\[
|d(c,a)-d(c,b)|\le d(a,b)\le d(c,a)+d(c,b).
\]
Choose any such point and call it $f(c)$.
Then all three edges of the initial triangle are realized isometrically.

\emph{Inductive step.}
Assume we have embedded the current $2$-tree $H'\subseteq H$ isometrically on edges:
\[
\|f(u)-f(v)\|=d(u,v)\qquad \forall \{u,v\}\in E(H').
\]
Let $H''$ be obtained from $H'$ by adding a new vertex $x$ adjacent to both endpoints $u,v$ of some existing edge
$\{u,v\}\in E(H')$.
We must choose $f(x)\in V$ such that
\[
\|f(x)-f(u)\|=d(x,u)\qquad\text{and}\qquad \|f(x)-f(v)\|=d(x,v).
\]
Since $\|f(u)-f(v)\|=d(u,v)$ by the inductive hypothesis, the triangle inequalities for $(x,u,v)$ yield
\[
|d(x,u)-d(x,v)|\le d(u,v)\le d(x,u)+d(x,v).
\]
Hence the spheres $S(f(u),d(x,u))$ and $S(f(v),d(x,v))$ intersect by \cref{lem:sphere-intersection-normed-plane}.
Choose any point $p$ in the intersection and set $f(x)=p$.
Then the new edges $\{x,u\}$ and $\{x,v\}$ are realized isometrically, and all previously realized edges remain unchanged.

Proceeding through the construction of $H$ defines $f$ on all vertices so that every edge of $H$
has the correct norm length; since $E\subseteq E(H)$, the same holds for all edges of $G$.
\end{proof}

\section{Lower bounds for graphs of bounded maximum degree}
\label{sec:lower_bdd}

In light of the equilateral space example (\cref{example:equilateral}) and anticipating an open problem (\cref{q:aspect-ratio}), we will track the \emph{aspect ratio} of the metric spaces appearing in the statements of our results. 

\begin{definition}
    \label{def:aspect-ratio}
Let $(X,d_X)$ be a finite metric space. The aspect ratio of $X$ is defined to be the ratio of the maximum to minimum distance in $X$, i.e.~
\[A_X := \frac{\max_{x\neq y}d_X(x,y)}{\min_{x\neq y}d_X(x,y)}.\]
\end{definition}

Our main result is the following. 

\begin{theorem}
\label{thm:gen-lb}
Let $(X,d_X)$ be a metric space with $|X| = n$ and aspect ratio $A_X$. 

Suppose for a metric space $(Y,d_Y)$, there exists a real number $D \geq 1$ such that any embedding $f: (X,d_X) \to (Y,d_Y)$ has distortion at least $D$. Then, for any $\delta \in (0,1/100)$, there exists a metric space $(Z,d_Z)$ -- with $|Z| = O(n^2)$ and aspect ratio $A_Z = O(A_X\cdot D^2 \log{n}/\delta)$ -- and a graph $G = (Z,E)$ with maximum degree $3$ such that for any $h : (Z,d_Z) \to (Y,d_Y)$, the $G$-local distortion of $h$ is at least $D-\delta$.

Moreover, if $(X,d_X)$ is an $\ell_p$-metric space for $p \in [1,\infty]$, then $(Z,d_Z)$ can be taken to be an $\ell_p$ metric space as well. 

\end{theorem}

\paragraph{\bf Remark} We record several remarks about the optimality of this result. 
\begin{enumerate}
    \item The bound of $3$ on the maximum degree of $G$ is the smallest possible. In \cref{prop:ub-deg2}, we show that for any finite metric space $(X,d_X)$ and any graph $G = (X,E)$ of maximum degree at most $2$, there exists a $G$-local isometric embedding $f: X \to (\mb{R}^2, \|\cdot \|)$ for any norm $\|\cdot \|$. On the other hand, as discussed earlier, there are examples of $(X,d_X)$ such that any embedding into $\ell_2$ incurs distortion $\Omega(\log n)$. 

    \item The bound $|Z| = O(n^2)$ cannot be improved in general. Indeed, Indyk and Wagner \cite[Theorem~6.2]{indyk2017near} showed that $\Omega(n^2 \log (1/\eps))$ bits are required to sketch finite metric spaces on $n$ points for which all distances are between $1$ and $2$, up to distortion $(1+\eps)$ (see \cite[Section~2]{indyk2017near} for formal definitions). Apply the construction in the proof of our theorem with $D=1+2\eps$ and $\delta=\eps$, and round each edge length upward to the nearest integer power of $(1+\eps)$. Let $\widetilde d$ be the shortest-path metric induced by these rounded edge lengths, and take $(Y,d_Y)=(Z,\widetilde d)$ and $h=\operatorname{id}_Z$. Then $h$ is a $G$-local $(1+\eps)$-embedding, and the estimates in the proof show that $\widetilde d|_X$ approximates $d_X$ within distortion $(1+2\eps)$. Thus the rounded edge lengths give a sketch for $(X,d_X)$ with distortion $(1+2\eps)$. Taking, say, $\eps=n^{-10}$, the bounded-degree graph $G$, the distinguished copy of $X$, and the rounded edge lengths can all be encoded using $O(|Z|\log n)$ bits (unless $|Z|\geq n^2$, in which case there is nothing to prove). This contradicts the lower bound in \cite{indyk2017near}, applied with accuracy parameter $2\eps$, unless $|Z| = \Omega(n^2)$.   

    \item On the other hand, we do not know of any obstruction showing that the bounds on the aspect ratio $A_Z$ are not improvable. In particular, it would be interesting if there is a construction with both $\Delta(G)$ and $A_Z$ depending only on $A_X, D, \delta$ (in particular, independent of $n$). See \cref{q:aspect-ratio}.
\end{enumerate}

\begin{proof}[Proof of \cref{thm:gen-lb}]
    Since $|X| = n$, it follows that $(X,d_X)$ can be isometrically embedded in $\ell_{\infty}^n$ (e.g., using the Fr\'echet embedding, see \cite{matouvsek2013lecture}); in particular, by composing with the canonical embedding $\ell_{\infty}^{n} \to \ell_{\infty}^{2n-2}$ (for $n\geq 2$), we may identify $X$ with a subset of points in $\ell_{\infty}^{2n-2}$. 

    Let $e_1,\dots, { e_{2n-2}}$ denote the standard basis of {$\ell_{\infty}^{2n-2}$}  and for {$i \in [2n-2]$}, let $v_i := \alpha \cdot e_i$, where $\alpha$ will be chosen later. Consider the following set of points in {$\ell_{\infty}^{2n-2}$}, 
    \[Z = X \cup \{x + v_i: x \in X, i \in [2n-2]\}\]
    and let $(Z,d_Z)$ be the metric space induced by the $\ell_\infty$ metric on these points. Note that $|Z| = O(n^2)$, as claimed. 

    We will construct a graph, $G = (Z, E)$ such that $(Z,d_Z)$ and $G$ will satisfy the conclusion of the theorem. The edges of $G$ are the union of the following two sets $E_1$ and $E_2$:

    $E_1$ is the union of $|X|$ disjoint complete binary trees, rooted at $x \in X$, where the different binary trees span the points $\{x, x+v_1,\dots, x+v_{2n-2}\}$ for different choices of $x \in X$. Note that each tree has $\lceil (2n-1)/2 \rceil = n$ leaves; we will assume that these leaves are $x+v_1,\dots, x+v_n$. These trees have height at most $L:=\lceil\log_2 n\rceil$. Note that in $E_1$, the roots have degree $2$, the leaves have degree $1$, and all other vertices have degree $3$.  

    We now describe the set $E_2$. Let $K_n$ denote the complete graph on $X$, and set $q=n-1$ if $n$ is even and $q=n$ if $n$ is odd. By Baranyai's Theorem \cite{Baranyai74}, the edges of $K_n$ can be decomposed into a disjoint union of $q$ matchings $M_1,\dots, M_q$ (perfect if $n$ is even and near-perfect if $n$ is odd; see \cref{fig:Baranyai} for an illustration). 

        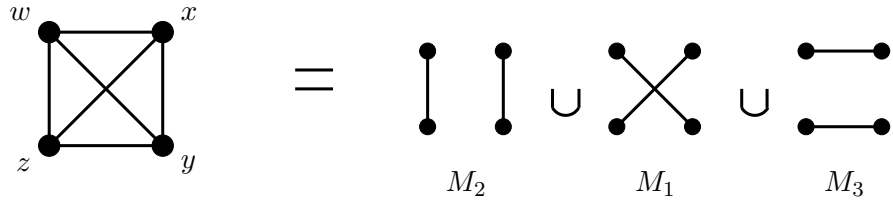
\begin{figure}[h]
\begin{tikzpicture}

    \draw[fill=black] (-9, 3) circle (4pt);
     \node[left=3pt] at (-9, 3.25) {$w$};

    \draw[fill=black] (-7.5,3) circle (4pt);
    \node[right=3pt] at (-7.5,3.25) {$x$};

    \draw[fill=black] (-9,1.5) circle (4pt);
    \node[left=3pt] at (-9,1.25) {$z$};

    \draw[fill=black] (-7.5,1.5) circle (4pt);
    \node[right=3pt] at (-7.5,1.25) {$y$};

     \draw[fill=black] (-4, 2.75) circle (3pt);
     \draw[fill=black] (-3, 2.75) circle (3pt);
     \draw[fill=black] (-3, 1.75) circle (3pt);
     \draw[fill=black] (-4, 1.75) circle (3pt);
     \node[right=3pt] at (-4,1) {$M_2$};

      \draw[fill=black] (-1.5, 2.75) circle (3pt);
     \draw[fill=black] (-.5, 2.75) circle (3pt);
     \draw[fill=black] (-.5, 1.75) circle (3pt);
     \draw[fill=black] (-1.5, 1.75) circle (3pt);
     \node[right=3pt] at (-1.5,1) {$M_1$};

     \draw[fill=black] (1, 2.75) circle (3pt);
     \draw[fill=black] (2, 2.75) circle (3pt);
     \draw[fill=black] (2, 1.75) circle (3pt);
     \draw[fill=black] (1, 1.75) circle (3pt);
     \node[right=3pt] at (1,1) {$M_3$};

    \draw[line width=0.4mm] (-5.25, 2.5) -- (-5.75, 2.5);

    \draw[line width=0.4mm] (-5.25, 2.25) -- (-5.75, 2.25);

    \draw[line width=0.4mm] (-4, 2.75) -- (-4, 1.75);
     \draw[line width=0.4mm] (-3, 2.75) -- (-3, 1.75);

    \draw[line width=0.4mm] (-1.5, 2.75) -- (-.5, 1.75);
     \draw[line width=0.4mm] (-1.5, 1.75) -- (-.5, 2.75);

     \draw[line width=0.4mm] (1, 2.75) -- (2, 2.75);
     \draw[line width=0.4mm] (2, 1.75) -- (1, 1.75);

    \draw[line width=0.4mm] (-9, 3) -- (-7.5, 3);
     \draw[line width=0.4mm] (-9, 3) -- (-7.5, 1.5);
     \draw[line width=0.4mm] (-9, 3) -- (-9, 1.5);
     \draw[line width=0.4mm] (-7.5, 3) -- (-7.5, 1.5);
     \draw[line width=0.4mm] (-7.5, 3) -- (-9, 1.5);
     \draw[line width=0.4mm] (-9, 1.5) -- (-7.5, 1.5);

    \draw[line width=0.4mm] (-2.35,2) arc (-150:-30:0.2);
    \draw[line width=0.4mm] (.15,2) arc (-150:-30:0.2);
    \draw[line width=0.4mm] (-2.35, 2.25) -- (-2.35, 2);
     \draw[line width=0.4mm] (-2, 2.25) -- (-2, 2);
     \draw[line width=0.4mm] (.15, 2.25) -- (.15, 2);
     \draw[line width=0.4mm] (.485, 2.25) -- (.485, 2); 
\end{tikzpicture}

 \caption{$K_4$ can be decomposed into three disjoint perfect matchings $M_1, M_2,$ and $M_3$.}
 \label{fig:Baranyai}
\end{figure}

   We define  
    \[E_2 = \{\{x+v_j, y+v_j\}: \{x,y\} \in M_j \text{ for some } j\in [q]\}.\]
    In words, for each pair of points $\{x,y\}$ in $X$, let $M_j$ be the unique matching in which the edge $\{x,y\}$ appears. Then, $E_2$ contains an edge between $x+v_j$ and $y+v_j$. Clearly, $E_2$ is itself a matching, i.e.~has maximum degree $1$. Also, since $x+v_j$ is a leaf in $E_1$ for $j \in [n]$, it follows that $G = (Z, E)$ with $E = E_1 \cup E_2$ has maximum degree 3. See \cref{fig:binary_tree} for an illustration.  
    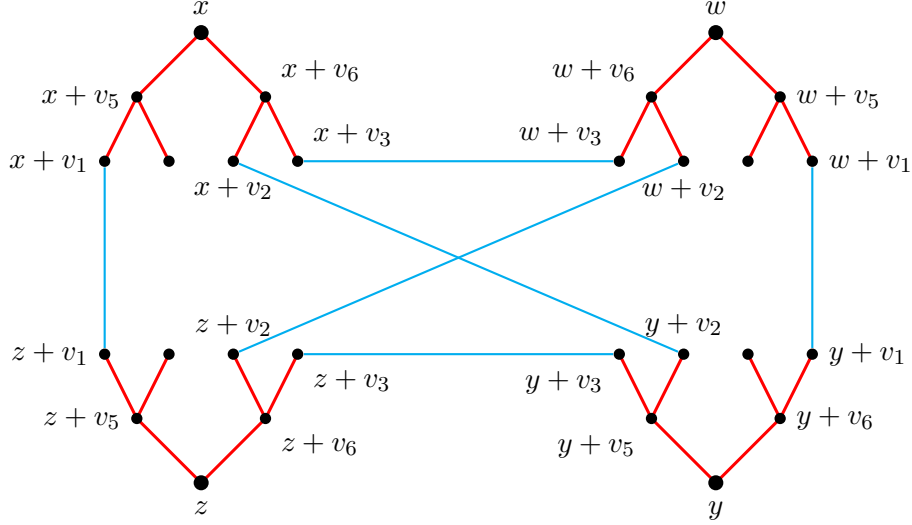
\begin{figure}[t]
\begin{tikzpicture}[scale=0.85]

\begin{scope}[xshift=4cm]

  \draw[red, line width=0.4mm] (-10, 3) -- (-11, 2);
  \draw[red, line width=0.4mm] (-10, 3) -- (-9, 2);
  \draw[red, line width=0.4mm] (-11, 2) -- (-11.5, 1);
  \draw[red, line width=0.4mm] (-11, 2) -- (-10.5, 1);
  \draw[red, line width=0.4mm] (-9, 2) -- (-9.5, 1);
  \draw[red, line width=0.4mm] (-9, 2) -- (-8.5, 1);

  \node[fill=black,circle,inner sep=2pt,label=above:$x$] (xroot) at (-10,3) {};
  \node[fill=black,circle,inner sep=1.5pt,label=left:$x+v_5$] at (-11,2) {};
  \node[fill=black,circle,inner sep=1.5pt,label={[label distance=0.5pt]above right :$x+v_6$}] at (-9,2) {};

  \node[fill=black,circle,inner sep=1.5pt,label=left:$x+v_1$] (xv3) at (-11.5,1) {};
  \node[fill=black,circle,inner sep=1.5pt] at (-10.5,1) {};
  \node[fill=black,circle,inner sep=1.5pt,label=below:$x+v_2$] (xv5) at (-9.5,1) {};
  \node[fill=black,circle,inner sep=1.5pt,label={[label distance=0.5pt]above right:$x+v_3$}] (xv6) at (-8.5,1) {};
\end{scope}

\begin{scope}[xshift=12cm]

  \draw[red, line width=0.4mm] (-10, 3) -- (-11, 2);
  \draw[red, line width=0.4mm] (-10, 3) -- (-9, 2);
  \draw[red, line width=0.4mm] (-11, 2) -- (-11.5, 1);
  \draw[red, line width=0.4mm] (-11, 2) -- (-10.5, 1);
  \draw[red, line width=0.4mm] (-9, 2) -- (-9.5, 1);
  \draw[red, line width=0.4mm] (-9, 2) -- (-8.5, 1);

  \node[fill=black,circle,inner sep=2pt,label=above:$w$] (wroot) at (-10,3) {};
  \node[fill=black,circle,inner sep=1.5pt,label={[label distance=0.5pt]above left :$w+v_6$}] at (-11,2) {};
  \node[fill=black,circle,inner sep=1.5pt,label=right:$w+v_5$] at (-9,2) {};

  \node[fill=black,circle,inner sep=1.5pt,label={[label distance=1pt]above left:$w+v_3$}] (wv6) at (-11.5,1) {};
  \node[fill=black,circle,inner sep=1.5pt,label={[label distance=0.5pt]below :$w+v_2$}]  (wv5) at (-10.5,1) {};
  \node[fill=black,circle,inner sep=1.5pt] at (-9.5,1) {};
  \node[fill=black,circle,inner sep=1.5pt,label=right:$w+v_1$] (wv3) at (-8.5,1) {};
\end{scope}

\begin{scope}[xshift=12cm,yshift=-1cm,yscale=-1]

  \draw[red, line width=0.4mm] (-10, 3) -- (-11, 2);
  \draw[red, line width=0.4mm] (-10, 3) -- (-9, 2);
  \draw[red, line width=0.4mm] (-11, 2) -- (-11.5, 1);
  \draw[red, line width=0.4mm] (-11, 2) -- (-10.5, 1);
  \draw[red, line width=0.4mm] (-9, 2) -- (-9.5, 1);
  \draw[red, line width=0.4mm] (-9, 2) -- (-8.5, 1);

  \node[fill=black,circle,inner sep=2pt,label=below:$y$] (yroot) at (-10,3) {};
  \node[fill=black,circle,inner sep=1.5pt,label={[label distance=1pt]below left:$y+v_5$}] at (-11,2) {};
  \node[fill=black,circle,inner sep=1.5pt,label=right:$y+v_6$] at (-9,2) {};

  \node[fill=black,circle,inner sep=1.5pt,label={[label distance=1pt]below left:$y+v_3$}] (yv6) at (-11.5,1) {};
  \node[fill=black,circle,inner sep=1.5pt,label={[label distance=0.5pt]above :$y+v_2$}] (yv5) at (-10.5,1) {};
  \node[fill=black,circle,inner sep=1.5pt] at (-9.5,1) {};
  \node[fill=black,circle,inner sep=1.5pt,label=right:$y+v_1$] (yv3) at (-8.5,1) {};
\end{scope}

\begin{scope}[xshift=4cm,yshift=-1cm,yscale=-1]

  \draw[red, line width=0.4mm] (-10, 3) -- (-11, 2);
  \draw[red, line width=0.4mm] (-10, 3) -- (-9, 2);
  \draw[red, line width=0.4mm] (-11, 2) -- (-11.5, 1);
  \draw[red, line width=0.4mm] (-11, 2) -- (-10.5, 1);
  \draw[red, line width=0.4mm] (-9, 2) -- (-9.5, 1);
  \draw[red, line width=0.4mm] (-9, 2) -- (-8.5, 1);

  \node[fill=black,circle,inner sep=2pt,label=below:$z$] (zroot) at (-10,3) {};
  \node[fill=black,circle,inner sep=1.5pt,label=left:$z+v_5$] at (-11,2) {};
  \node[fill=black,circle,inner sep=1.5pt,label={[label distance=0.5pt]below right :$z+v_6$}] at (-9,2) {};

  \node[fill=black,circle,inner sep=1.5pt,label=left:$z+v_1$] (zv3) at (-11.5,1) {};
  \node[fill=black,circle,inner sep=1.5pt] at (-10.5,1) {};
  \node[fill=black,circle,inner sep=1.5pt,label={[label distance=0.5pt]above :$z+v_2$}]  (zv5) at (-9.5,1) {};
  \node[fill=black,circle,inner sep=1.5pt,label={[label distance=0.5pt]below right :$z+v_3$}] (zv6) at (-8.5,1) {};
\end{scope}

\draw[cyan,thick] (wv3) -- (yv3);
\draw[cyan,thick] (xv3) -- (zv3);
\draw[cyan,thick] (wv6) -- (xv6);
\draw[cyan,thick] (zv6) -- (yv6);
\draw[cyan,thick] (wv5) -- (zv5);
\draw[cyan,thick] (xv5) -- (yv5);

\end{tikzpicture}

    \caption{The construction of $G = (Z, E)$ for $X = \{w,x,y,z\}$. The edges in $E_2$, corresponding to the decomposition in \cref{fig:Baranyai}, are depicted in blue.}
    \label{fig:binary_tree}
\end{figure}

  To complete our construction, it remains to specify our choice of $\alpha$. Let $\gamma := \min_{x\neq y}d_X(x,y)$ be the minimum distance between two distinct points of $X$. We set $\alpha = \delta \gamma/(8D^2 L)$. With this choice of parameters, the assertion about the aspect ratio of $(Z, d_Z)$ is immediate. 
  Additionally, we have the following. 

    \begin{claim}
        \label[claim]{claim:point-cloud}
   Suppose the map $h: (Z,d_Z) \to (Y,d_Y)$ is a $G$-local $D'$-embedding with $D' \leq D$ (i.e.~satisfies \cref{eq:distortion-local} with some $r > 0$ and $1\leq D'\leq D$). Then, for any $x \in X$ and $j\in[n]$, 
    \[d_Y(h(x), h(x+v_j)) \leq r\cdot \delta/8D\cdot \gamma\]
    \end{claim}
    \begin{proof}
        By construction, there is a path from $x$ to $x+v_j$ of length at most $L$ in $E$. The edges of the path are of the form $\{x, x+v_k\}$ or $\{x+v_k, x+v_\ell\}$. Denoting the points along the path by $x = y_1, y_2, \dots, y_{s} = x+v_j$, we have
        \begin{align*}
            d_Y(h(x), h(x+v_j))
            &\leq d_Y(h(y_1), h(y_2)) + \dots + d_Y(h(y_{s-1}), h(y_{s}))\\
            &\leq r\cdot D\cdot \left(d_{Z}(y_1, y_2) + \dots + d_{Z}(y_{s-1}, y_{s})\right)\\
            &\leq r\cdot D \cdot (s-1)\cdot \alpha\\ & \leq r\cdot \delta/8D\cdot \gamma;
        \end{align*}
        here, the first inequality is the triangle inequality; the second inequality uses that $h$ is a $G$-local $D'$-embedding with $D' \leq D$; the third inequality uses $d_{Z}(y_k, y_{k+1}) \leq \alpha$, which is true by construction; and the last inequality uses our setting of $\alpha$ along with $s-1 \leq L$. 
    \end{proof}
    With this claim in hand, we proceed with the proof of the theorem. Suppose for contradiction that there exists $h: (Z,d_Z) \to (Y,d_Y)$ which is a $G$-local $D'$-embedding with $D' < D-\delta$. Let $f = h|_X : (X,d_X) \to (Y,d_Y)$ denote the restriction of $h$ to $X$. We will show that the distortion of $f$ is strictly less than $D$, thereby contradicting our assumption. 

For $x, y \in X$, $x \neq y$, let $j\in [q]$ be such that $\{x+v_j, y+v_j\}\in E_2\subseteq E$; recall that such a value of $j$ is guaranteed to exist by construction. Then,
\begin{align*}
    d_Y(f(x), f(y)) &= d_Y(h(x), h(y))\\
    &\leq d_Y(h(x), h(x+v_j)) + d_Y(h(x + v_j), h(y+v_j)) + d_Y(h(y+v_j), h(y))\\
    &\leq r\cdot (D-\delta) \cdot d_{Z}(x+v_j, y+v_j) + 2\cdot r\cdot \delta/8D\cdot \gamma\\
    &= r\cdot (D-\delta) \cdot d_{Z}(x, y) + r\cdot \delta/4D\cdot \gamma\\
     &= r\cdot (D-\delta) \cdot d_{X}(x, y) + r\cdot \delta/4D\cdot \gamma \\
    &\leq r\cdot \left(D-\delta+\frac{\delta}{4D}\right)\cdot d_X(x,y);
\end{align*}
here, the second line is the triangle inequality; the third line follows from \cref{claim:point-cloud} and the assumption that $h$ is a $G$-local $(D-\delta)$-embedding; the fourth line follows since the metric on $Z$ is a norm; and the last line follows from the definition of $\gamma$. 

Similarly, we have
\begin{align*}
    d_Y(f(x), f(y)) &= d_Y(h(x), h(y))\\
    &\geq -d_Y(h(x), h(x+v_j)) + d_Y(h(x + v_j), h(y+v_j)) - d_Y(h(y+v_j), h(y))\\
    &\geq r\cdot d_{Z}(x + v_j, y +v_j) -  2\cdot r\cdot \delta/8D\cdot \gamma \\
    &= r\cdot d_X(x,y) - r\cdot \delta/4D\cdot \gamma \\
    &\geq r(1-\delta/4D)\cdot d_X(x,y).
\end{align*}

Combining the above two equations shows that the distortion of $f$ is bounded above by $\frac{D-\delta+\delta/(4D)}{1-\delta/(4D)}$, which is strictly less than $D$ since $D\geq 1$; this gives us the desired contradiction. 

Finally, for the ``moreover'' part, if $(X,d_X)$ is an $\ell_p$ metric space to start with, then we can view $X$ as a subset of points in $\ell_p^{N}$ for $N = \binom{n}{2}$ (see,~e.g.~\cite[Proposition~1.4.2]{matouvsek2013lecture}), then canonically as a subset of $\ell_p^{N+2n-2}$, and repeat the proof above with $v_i=(\alpha/2)e_{N+i}$. Since $\|v_i-v_j\|_p\leq\alpha$, all the preceding estimates, including the asserted aspect-ratio bound, remain valid up to absolute constants.
\end{proof}

\subsection{Graphs of maximum degree two} As remarked after the statement of \cref{thm:gen-lb}, one cannot prove the theorem with graphs of maximum degree $2$. This follows immediately from \cref{thm:tw2}.

\begin{proposition}
\label[proposition]{prop:ub-deg2}
Let $(X,d)$ be a finite metric space and let $G=(X,E)$ be a graph of maximum degree at most $2$.
Let $(\mathbb{R}^2,\|\cdot\|)$ be any two-dimensional normed space.
Then there exists a map $f:X\to\mathbb{R}^2$ such that
\[
\|f(x)-f(y)\| = d(x,y)\qquad \forall \{x,y\}\in E.
\]
\end{proposition}

\begin{proof}
If $\Delta(G)\le 2$, then every connected component of $G$ is a path or a cycle.
Paths have treewidth $1$ and cycles have treewidth $2$, hence $\tw(G)\le 2$.
The conclusion now follows immediately from \cref{thm:tw2}.
\end{proof}

\subsection{Applications} We conclude with two quick applications of \cref{thm:gen-lb}. First, as discussed in the introduction, it was shown by Abraham, Bartal, and Neiman \cite[Theorem~1]{abraham2009low} that any metric space on $n$ points can be embedded into $\ell_p^{O(e^p \log^2k)}$ with $k$-local distortion $O(\log k/p)$ (for any $k\leq n$ and $1\leq p \leq \log k$). We show that for graphically local embeddings, this fails in a strong sense: below, we provide an example of an $n$ point metric space and a graph $G$ of maximum degree $3$ such that any $G$-local embedding of this metric space into $\ell_2$ incurs $G$-local distortion $\Omega(\log n)$ (recall that $O(\log n)$ \emph{global} distortion is always achievable by Bourgain's theorem). We remark that a similar argument works for all $\ell_p$ spaces with fixed $p\geq 1$. 

\begin{corollary}
\label[corollary]{cor:l1-embedding}
For any integer $n\geq 2$, there exists a metric space $(Z,d_Z)$ with $|Z| = \Theta(n)$ and a graph $G = (Z,E)$ of maximum degree $3$ such that any $G$-local embedding of $(Z, d_Z)$ into $\ell_2$ has $G$-local distortion $\Omega(\log n)$.
\end{corollary}

\begin{proof}
It is well known that there exists a metric space $(X,d_X)$ on $\Theta(\sqrt{n})$ points such that any embedding of $(X,d_X)$ into $\ell_2$ incurs distortion $\Omega(\log n)$; for instance, one can take $(X,d_X)$ to be the graph metric on a sufficiently good expander with $\Theta(\sqrt{n})$ vertices (see,~e.g.~\cite[Theorem~3.5.3]{matouvsek2013lecture}).

Let $N:=|X|=\Theta(\sqrt{n})$, and let $D := c\log N$ be the corresponding global lower bound (for some absolute $c>0$). Apply \cref{thm:gen-lb} with this $D$ and, say, $\delta:=1/200$. The output is a metric space $(Z,d_Z)$ with $|Z|=O(N^2)=O(n)$ and a graph $G$ of maximum degree $3$ such that any $G$-local embedding into $\ell_2$ has $G$-local distortion at least $D-\delta=\Omega(\log n)$ (absorbing constant factors). By construction, $|Z|=N(2N-1)=\Theta(n)$.
\end{proof}

Next we show that, in general, $G$-local dimension reduction is no easier than global dimension reduction, even for graphs of maximum degree $3$. Previously, such a construction was only known under the additional restriction that the embedding is noncontracting (recall \cite[Theorem~11]{schechtman2009lower} due to Schechtman and Shraibman). 

\begin{corollary}
\label[corollary]{cor:local-JL}
For any fixed $\eta>0$, any integers $n,d \geq 2$, and any $\varepsilon \in ( n^{\eta}/\sqrt{\min(\sqrt{n},d)}, 1/100)$, there exists a set of points $Z \subseteq \mb{R}^d$ and a graph $G = (Z, E)$ of maximum degree $3$ such that the following hold:
\begin{itemize}
\item $|Z| = O(n)$; 
\item the aspect ratio of the Euclidean metric on $Z$ is $O(\log n\cdot \varepsilon^{-2})$;
\item for any map $f: Z \to \mb{R}^m$ satisfying
\[\|x-y\|_2 \leq \|f(x)-f(y)\|_2 \leq (1+\varepsilon)\|x-y\|_2 \qquad \forall \{x,y\} \in E,\]
we must have
\[m = \Omega_\eta(\varepsilon^{-2}\log n).\]  
\end{itemize}
\end{corollary}

 Our result leaves open the following questions. 
\begin{question}
For global dimension reduction, Alon and Klartag \cite{alon2017optimal} prove the lower bound
\[
\Omega\left(\min\left\{d,n,\frac{\log(2+\eps^2n)}{\eps^2}\right\}\right)
\]
throughout the full range of parameters. Is the same lower bound also true for $G$-local dimension reduction? In \cref{cor:local-JL}, we obtain this lower bound only when $\eps$ is polynomially larger than $1/\sqrt{\min(\sqrt n,d)}$.
\end{question}

\begin{question}
\label[question]{q:aspect-ratio}
   More interestingly, do there exist lower bound instances with bounded aspect ratio? Or is it the case that for every Euclidean metric space $X$ on $n$ points whose distances are all between $1$ and $\Phi$, and for every graph $G = (X,E)$ with maximum degree $\Delta$, for every $\eps > 0$, there exists a $G$-local embedding of $X$ into $\ell_2^{m}$ with distortion $(1+\eps)$ and with $m = f(\eps,\Delta, \Phi)$? Observe that this is true for the special case $\Phi = 1$ (\cref{example:equilateral}). Also note that, since storing all distances appearing in $G$ to $(1+\eps)$ multiplicative error only requires $O(n\Delta\log(1/\eps)+ n\Delta\log\log\Phi)$ bits, the existence of such an embedding cannot be ruled out simply by sketching arguments (for the JL lemma, the tight lower bound can be obtained from lower bounds on sketching, see \cite{alon2017optimal}).    
\end{question}

\begin{proof}[Proof of \cref{cor:local-JL}]
Fix $\eta>0$ and let $n' := \lfloor\sqrt{n}\rfloor$. The assertion for bounded $n$ is trivial after decreasing the implicit constant in the conclusion, so we may assume that $n$ is sufficiently large in terms of $\eta$. Our assumption on $\eps$ then implies
\[
\eps > \frac{\log^{0.51}n'}{\sqrt{\min(n',d)}}.
\]
Applying the main result of Larsen and Nelson \cite{LarsenNelson17} with parameter $4\eps$, we obtain a set of points $X\subseteq \mb{R}^d$ of size $n'$ and aspect ratio $O(\eps^{-1})$ such that any embedding of $(X,d_X)$ into $\mb{R}^m$ with distortion $(1+4\eps)$ must satisfy
\[
m = \Omega\left(\eps^{-2}\log(\eps^2n')\right)
  = \Omega_{\eta}\left(\eps^{-2}\log n\right).
\]
Here, the last equality follows since the assumed lower bound on $\eps$ gives $\eps^2n'=\Omega(n^{2\eta})$.

Apply the construction in the proof of the ``moreover'' part of \cref{thm:gen-lb} to $(X,d_X)$ with $D=1+4\eps$ and $\delta=\eps$. This gives a Euclidean metric space $(Z_0,d_0)$ with $|Z_0|=O(n)$ and aspect ratio $O(\eps^{-2}\log n)$, and a graph $G_0=(Z_0,E_0)$ of maximum degree $3$, such that every $G_0$-local embedding of $(Z_0,d_0)$ into $\mb{R}^m$ has distortion at least $1+3\eps$ whenever $m<c_{\eta}\eps^{-2}\log n$, for a sufficiently small constant $c_{\eta}>0$.

By the Johnson-Lindenstrauss lemma, there is an embedding $\phi:Z_0\to\mb{R}^r$ with distortion at most $1+\eps$, where
\[
r=O(\eps^{-2}\log |Z_0|)=O(\eps^{-2}\log n)\leq d.
\]
The last inequality follows from our assumption on $\eps$, provided $n$ is sufficiently large in terms of $\eta$. Identifying $\mb{R}^r$ with a subspace of $\mb{R}^d$, set $Z=\phi(Z_0)$ and let $G=(Z,E)$ be the graph corresponding to $G_0$. Then $|Z|=O(n)$ and the aspect ratio of $Z$ is $O(\eps^{-2}\log n)$.

Suppose that $f:Z\to\mb{R}^m$ satisfies the inequalities in the statement. Then $f\circ\phi$ is a $G_0$-local embedding of $(Z_0,d_0)$ into $\mb{R}^m$ with distortion at most $(1+\eps)^2<1+3\eps$. This is impossible when $m<c_{\eta}\eps^{-2}\log n$, completing the proof.
\end{proof}

\bibliographystyle{amsplain0.bst}
\bibliography{main.bib}

\end{document}